\documentclass[oneside, twocolumn, conference]{IEEEtran}
\IEEEoverridecommandlockouts

\usepackage[utf8]{inputenc}
\usepackage[T1]{fontenc}

\usepackage{amsmath}
\usepackage{mathtools}

\usepackage{balance}

\usepackage{caption}
\usepackage{subcaption}

\usepackage{amsthm}
\theoremstyle{definition}
\newtheorem{example}{Example}
\newtheorem{lemma}{Lemma}

\usepackage[noadjust]{cite}

\usepackage[noend]{algpseudocode}
\usepackage{algorithm}

\algdef{SE}[DOWHILE]{Do}{doWhile}{\algorithmicdo}[1]{\algorithmicwhile\ #1}%

\title{Divide (CPU Load) and Conquer: Semi-Flexible Cloud Resource Allocation
\thanks{Preprint of the paper accepted at the 22nd IEEE/ACM International Symposium on Cluster, Cloud and Internet Computing (CCGrid 2022), Taormina (Messina), Italy, 2022}
}

\author{
    \IEEEauthorblockN{%
        Bartłomiej Przybylski\IEEEauthorrefmark{1}, Paweł Żuk\IEEEauthorrefmark{2}, Krzysztof Rzadca\IEEEauthorrefmark{3}}
    \IEEEauthorblockA{%
        \emph{Institute of Informatics, University of Warsaw}\\
        Warsaw, Poland\\
        Email: \IEEEauthorrefmark{1}bap@mimuw.edu.pl, \IEEEauthorrefmark{2}p.zuk@mimuw.edu.pl, \IEEEauthorrefmark{3}krzadca@mimuw.edu.pl}
}

\usepackage{lipsum}

\usepackage{xcolor}

\usepackage{amsmath}
\usepackage{amssymb}
\newcommand\ja{\bigstar}
\newcommand\jb{\blacklozenge}
\newcommand\jc{\blacktriangledown}

\begin{document}

\maketitle

\begin{abstract}
Cloud resource management is often modeled by two-dimensional bin packing with a set of items that correspond to tasks having fixed CPU and memory requirements. However, applications running in clouds are much more flexible: modern frameworks allow to (horizontally) scale a single application to dozens, even hundreds of instances; and then the load balancer can precisely divide the workload between them.

We analyze a model that captures this (semi)-flexibility of cloud resource management. Each cloud application is characterized by its memory footprint and its momentary CPU load. Combining the scheduler and the load balancer, the resource manager decides how many instances of each application will be created and how the CPU load will be balanced between them. In contrast to the divisible load model, each instance of the application requires a certain amount of memory, independent of the number of instances. Thus, the resource manager effectively trades additional memory for more evenly balanced load. 

We study two objectives: the bin-packing-like minimization of the number of machines used; and the makespan-like minimization of the maximum load among all the machines. We prove NP-hardness of the general problems, but also propose polynomial-time exact algorithms for boundary special cases. Notably, we show that (semi)-flexibility may result in reducing the required number of machines by a tight factor of $2-\varepsilon$. For the general case, we propose heuristics that we validate by simulation on instances derived from the Azure trace.
\end{abstract}


\section{Introduction}
Borg \cite{Borg43438,Borg49065}, Resource Central \cite{RCentral3132772} and other cloud resource managers play a balancing act between the global efficiency --- packing the machines as densely as possible --- and serving quality --- ensuring that individual machines are rarely, if ever, overloaded. A typical cloud application uses two additional layers: (1) a (horizontal) autoscaler (e.g. \cite{rzadca2020autopilot}) that adds or removes instances in response to long-term changes in the application load; and (2) a load balancer (e.g. \cite{patel2013ananta}) that dynamically assigns end-user queries to instances for a shorter-term balance. These two layers typically issue requests to, rather than fully coordinate with, the resource manager.

In this paper, we show that --- by integrating the autoscaler, the resource manager (scheduler) and the load balancer --- the cloud may use resources more efficiently. 
In particular, we coordinate setting the number of instances (autoscaling), their placement on machines (scheduling) and the allocation of load to individual instances. Unlike the standard load balancing, our proposed method balances the load taking into account all the machine's assigned instances (and thus, all the applications).

The contemporary cloud software stack has enormous engineering complexity. Thus, to show the impact of our ideas, instead of a system study which would most probably be infeasible, we take a formal, algorithmic approach based on classic scheduling (which we additionally complement with simulations). This formal approach enables us to show trends and qualitative differences, but it requires some reasonable modeling of the problem. 
As in Google's Borg~\cite{Borg49065},
we focus on two key resources: the operational memory (RAM) and the computational power (CPU).
We model a system hosting multiple \emph{applications}, each processing a certain \emph{load}. 
This load is distributed between the application's \emph{instances}.
For example, in Function as a Service (FaaS), a single function corresponds to our application. The end-user-driven invocations of this function (a stream of HTTP requests) is the application's load. Any invocation can be processed by any machine that has initialized this function (In OpenWhisk \cite{OpenWhisk}, 
the invoker module dynamically claims the invocation from a global queue).
As another example, our application corresponds to a single serving job in the Google's Borg model \cite{Borg43438,Borg49065}; its instances correspond to the job's tasks; and a cluster-level load balancer assigns queries to these tasks.
Finally, to derive tangible formal results, we assume that the memory requirement of an instance does not depend on the load processed by this instance.
While we have no data to back up this assumption, various application classes should behave according to that model, with the memory requirements dominated by the software stack (libraries, etc.), pre-loaded datasets, or static data structures.
We stress that this is not a core assumption (as it could be easily extended to, e.g., a linear function), but rather a standard modeling step that allows us to demonstrate qualitative results.

In this formal model, we analyze two natural combinatorial optimization problems: (1) bin-packing-like minimization of the number of used machines; (2)  makespan-like minimization of the maximum load of any machine.
Bin-packing models applications' requirements as hard constraints. 
This corresponds to, e.g. packing high-priority jobs by their limits in Borg; or, in an IaaS provider, packing VMs by their sizes (as requested by customers) and maintaining strict SLOs with no overcommitment.
In contrast, the makespan-like approach explores the different nature of these resources. A memory requirement of an application cannot be (easily) compressed or throttled (cloud providers do not swap memory to disk)\cite{Borg43438}. Unlike memory, the CPU is compressible --- it can be dynamically throttled. Of course, when throttled, the application slows down, which is tolerable for batch, while should be avoided for serving applications. Thus, the makespan-like approach minimizes the load of the maximally-loaded machine, corresponding to minimization of the maximal throttling. (The ``makespan'' is a metaphor: we assume that applications are executed concurrently.)

The detailed contributions of this paper are the following:
\begin{itemize}
\item We present a combinatorial optimization model of cloud application allocation with load balancing and fixed memory requirements. We formulate two general optimization problems: packing and balancing (Sect.~\ref{sec:pd}).
\item We prove NP-Hardness for both models in the general case. We also show optimal polynomial algorithms for equal requirements. (Sect.~\ref{sec:min-cpu-usage}--\ref{sec:min-cpu-number}).
\item We propose heuristics for the packing problem which take into account applications' semi-flexibility (Sect.~\ref{sec:heur}).
\item We simulate heuristics with instances based on Azure Public Dataset V2 \cite{RCentral3132772} (Sect.~\ref{sec:experiments}). 
\end{itemize}

\section{Problem description}
\label{sec:pd}

In this section, we formally define the optimization problem of cloud application allocation with memory requirements and CPU load balancing as a general integer linear program (ILP). We follow the classic scheduling notation~\cite{brucker1999scheduling}. 

Let us be given $m$ identical machines, each having a memory capacity of $Q \in \mathbb{Z}_+$ (measured in, e.g., bytes) and the CPU capacity of $P \in \mathbb{Z}_+$ (measured in, e.g., vCPUs or Borg's Normalized CPUs~\cite{Borg43438,Borg49065}). 
We assume the machines are identical as cloud providers usually manage a few large groups of homogeneous machines (e.g.\ 4 machine types cover 98\% of machines of a Google's 10,000-machine cluster~\cite{Borg43438} --- we thus solve a separate instance for each of these 4 large groups). Similarly, if a system uses VMs rented from an IaaS provider, it is natural to use a Managed Instance Group that requires all VMs to have the same instance type.

Let us also be given $n$ \emph{applications}. 
Contemporary cloud applications are usually horizontally-scalable: 
multiple \emph{instances} of the same application, placed on multiple machines, jointly process the load of the application (e.g., for serving applications, each instance processes a fraction of the stream of incoming requests).
In cloud, this mechanism is additionally used to increase reliability (e.g.: 3 or 5 always-on instances). While such lower bounds on the number of active instances can be easily incorporated into our approach, they are mostly orthogonal to our results (so we do not discuss them further).

There is a cost of maintaining multiple instances, however: each instance of the $i$-th application has its own integer memory demand of $0 < q_i \leq Q$ (measured in the same unit as the memory capacity). For example, if the $i$-th application is placed on two machines, it uses a total of $2q_i$ units of memory: $q_i$ units on the first and $q_i$ units on the second machine. We assume that, for the $i$-th application, $q_i$ is constant --- in particular, unrelated to the load assigned to the instance. This corresponds to memory requirements dominated by the software stack (libraries etc.), or the dataset, rather than dynamically changing with the processing load. Our model can be extended to memory requirement being a (perhaps linear) function of the load assigned, but in this paper we prefer to keep our model simple and the memory requirements constant in order to prove formal results and show qualitative conclusions.

One of our goals is to illustrate how much we can improve the utilization of the whole cluster by maintaining multiple instances. We will test our \emph{multi-instanced models} against the standard scheduling models, later called \emph{single-instanced}.

Moreover, we know $p_i$, the total load that needs to be processed by application $i$. The load $p_i > 0$ is expressed as the number vCPUs it requires (we use the same metric for the load and the capacity following Borg~\cite{Borg43438,Borg49065}). In the single-instanced model, $p_i$ units of the vCPU capacity on a single machine need to be reserved for the sole instance.
In the multi-instanced model, the load balancer freely divides $p_i$ between the application's instances as long as the reserved CPU capacity sums up to at least $p_i$. The amount of computation assigned to a particular instance may be fractional, too. For many serving applications, the load consists of a huge number of relatively tiny requests (single API calls or function invocations). If the total QPS is in thousands, load balancer decisions can be reasonably approximated by fractional assignments.

We assume a classic, off-line and clairvoyant model with $p_i$ and $q_i$ known in advance, a common approach in cloud resource management research. The load $p_i$ does not correspond to processing time, but to the total load of the $i$-th application. When customers deploy their applications in cloud, they are commonly required to upper-bound the total number of vCPUs and memory --- and their application is then allocated based on these given upper-bounds. Additionally, serving applications are usually long-running: steady-state vCPU and memory requirements can be precisely estimated based on relatively simple models using historical trends~\cite{rzadca2020autopilot}. 

Our aim is to assign applications to machines
in such a way that all the incoming requests can be processed and that memory used on each machine does not exceed its capacity. As---in this model---it makes no sense to place two instances of the same application on the same machine (if such two instances were merged, they would process the same load using half the memory), we will use a 0-1 variable $x_{ij}$ to determine whether the instance of the $i$-th application is placed on the $j$-th machine, or not. If positive, the $p_{ij}$ variable will determine the total vCPU capacity reserved for the $i$-th application on the $j$-th machine. Furthermore, $p_{ij}/p_i$ corresponds to the share of the whole traffic of the $i$-th application routed by a load balancer to the $j$-th machine. We thus always want the following constraints to be satisfied:

\begin{itemize}
    \item As we do not overcommit memory on any machine, the memory utilization of all the instances placed on the $j$-th machine does not exceed machine's capacity $Q$, i.e.
        \begin{equation}
            \sum_i x_{ij}q_i \leq Q, \quad\text{for each $j$};
            \label{eq:memory-constraint}
        \end{equation}
    \item The total vCPU capacity reserved for the $i$-th application is greater or equal to the application's load $p_i$, i.e.
        \begin{equation}
            \sum_j x_{ij}p_{ij} \geq p_i, \quad\text{for each $i$};
            \label{eq:cpu-lower-constraint}
        \end{equation}
    \item Assignments $x_{ij}$ are binary, i.e.
        \begin{equation}
            x_{ij} \in \{0, 1\}, \quad\text{for each pair of $i$ and $j$}.
        \end{equation}
\end{itemize}

In Fig.~\ref{fig:exas}, we present an example allocation of three applications, $\{\ja, \jb, \jc \}$, to one machine. Total memory used by these three instances is equal to $q_{\ja} + q_{\jb} + q_{\jc} \leq Q$, so the constraint \eqref{eq:memory-constraint} is not violated on this machine. At the same time, the total CPU used slightly exceeds the value of $P$, i.e. $p_{\ja j} + p_{\jb j} + p_{\jc j} > P$ (so far we have not introduced machine-level constraints on the vCPU load). In fact, applications $\ja$ and $\jc$ are assigned as much vCPU capacity as they require, i.e. $p_{\ja j} = p_\ja$ and $p_{\jc j} = p_\jc$. However, the vCPU capacity assigned to application $\jb$ is strictly less than its whole demand, i.e. $p_{\jb j} < p_\jb$. Thus, another instance of application $\jb$ has to be allocated on at least one other machine, as otherwise constraint \eqref{eq:cpu-lower-constraint} will not be satisfied.

\begin{figure}
    \centering
    \includegraphics[width=\columnwidth]{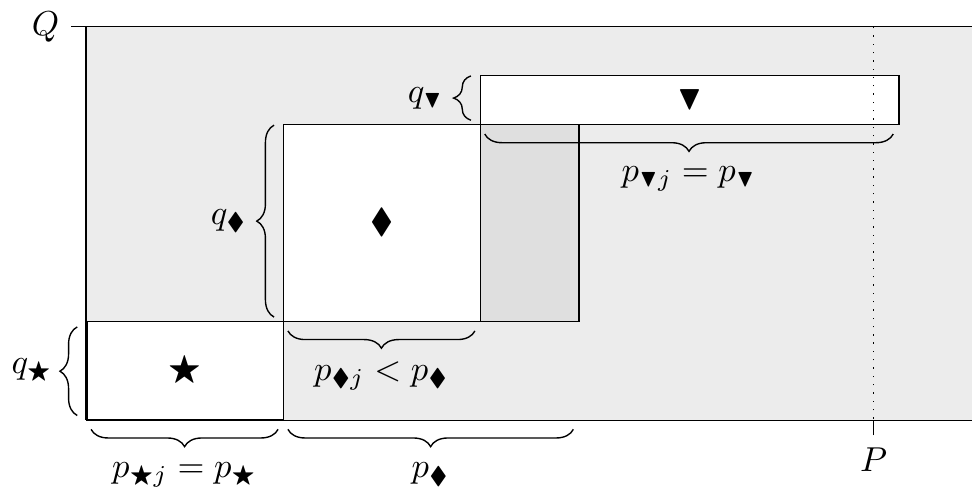}
    \caption{An example assignment of three applications to a single machine $j$ in the multi-instanced model. Memory capacity $Q$ (solid horizontal line) cannot be exceeded due to constraint~\eqref{eq:memory-constraint}, while vCPU capacity $P$ (dotted vertical line) can.}
    \label{fig:exas}
\end{figure}

We consider two natural optimization objectives.
First, given a fixed pool of homogenous machines we assign applications and workloads to machines in such a way that the maximum vCPU usage among all the machines is minimized:
\begin{equation}
    \min \max_j \sum_i x_{ij}p_{ij} \text{.}
\end{equation}
This objective models a cloud resource manager such as a single BorgMaster~(\cite{Borg43438,Borg49065}) that allocates load on a single cluster of physical machines.

Second, given a set of applications and a constraint on the maximum vCPU usage on each machine, we minimize the number of used machines.
This models a large customer minimizing the number of rented VMs while maintaining applications' SLOs. As formalization of this objective needs additional notation, we defer it to Sec.~\ref{sec:min-cpu-number}.

\section{Optimal Balancing: Minimizing the Maximum CPU usage}
\label{sec:min-cpu-usage}
In the optimal balancing problem, $\min \max_j \sum_i x_{ij}p_{ij}$, we start by analyzing polynomially-solvable special cases of common CPU and memory requirements; we then proceed to prove NP-hardness for arbitrary CPU requirements (with unit memory requirements); and arbitrary memory requirements (with unit CPU requirements).

\subsection{Common CPU and memory requirements}
\label{subsec:minmaxcpu-qp}
We now assume that all applications have the same CPU requirement of $p$ and the same memory requirement of $q$.
Although it would seem that it makes no sense to have multiple instances of any application, especially if additionally $p = 1$, it is not true, as the following example shows.

\begin{example}
Let us consider a two-machine environment ($m = 2$) where $Q = 2$, and three different applications such that $p_i = q_i = 1$ for $i \in \{\ja, \jb, \jc\}$. If one assigns applications $\ja$ and $\jb$ to the first machine and application $\jc$ to the second machine, then $\max_j \sum_i x_{ij}p_{ij} = 2$ (the maximum is reached on the first machine, see Fig.~\ref{fig:ex-assignment}(a)). However, if one assigns application $\ja$ to the first machine, application $\jc$ to the second machine, and application $\jb$ to \emph{both} the first and the second machine in such a way that $p_{\ja 1} = p_{\ja 2} = 0.5$, then $\max_j \sum_i x_{ij}p_{ij} = 1.5$ (see Fig.~\ref{fig:ex-assignment}(b)) which is the lower bound on the maximum CPU usage (as $\sum_i p_i/m = 1.5$).

\begin{figure}
    \begin{subfigure}[b]{0.5\textwidth}
        \centering
        \includegraphics[scale=0.85]{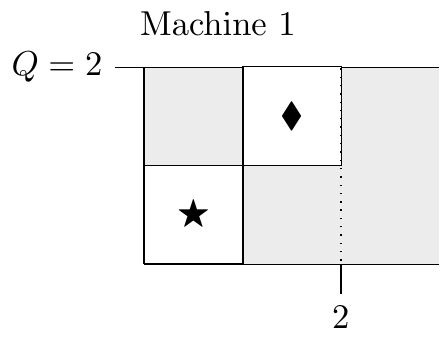}
        \includegraphics[scale=0.85]{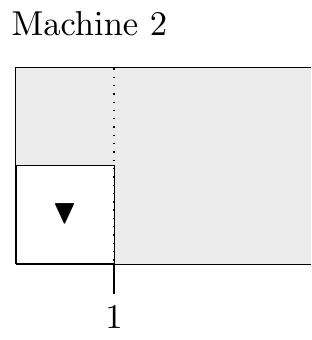}
        \caption{All the applications have a single instance}
    \end{subfigure}
    \hspace{1em}
    \begin{subfigure}[b]{0.5\textwidth}
        \centering
        \vspace*{0.5em}
        \includegraphics[scale=0.85]{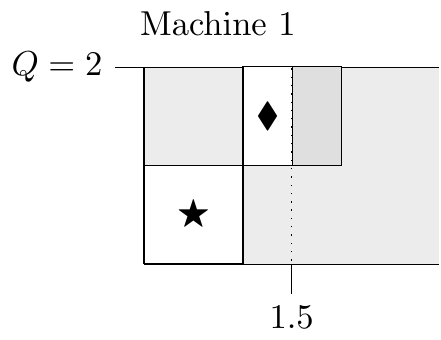}
        \includegraphics[scale=0.85]{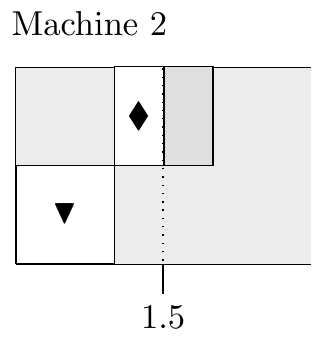}
        \caption{Application $2$ has two instances}
    \end{subfigure}
    \caption{Example assignment of three applications to two machines.}
    \label{fig:ex-assignment}
\end{figure}
\end{example} 

Alg.~\ref{alg:min-maxp-pq} shows how to solve this problem for any $q$ and $Q$. 
To simplify its description, we assume now that $q=1$ (it can be done without loss of generality).
Notice that $m\cdot Q \geq n\cdot q = n$ must hold (at least one instance of each application must be placed on some machine).
Otherwise, there would exist no feasible assignment of applications to machines. As $m$, $n$ and $Q$ are integers, it must also hold that $Q \geq \lceil n/m \rceil$. 
Be reminded that in this special case we assume that each application requires a total of $p$ CPU units. 
If $n \bmod m = 0$, then the optimal solution can be obtained by assigning any $n/m$ applications to each of the machines, without adding multiple instances. This solution leads to the optimal value of $\max_j \sum_i x_{ij}p_{ij} = pn/m.$ Now, assume that $n \bmod m > 0$ (which entails $\lceil n/m \rceil > n/m > \lfloor n/m \rfloor$). If $Q > \lceil n/m \rceil$, then at least $\lfloor n/m \rfloor + 2$ instances can be placed on each machine. If so, the optimal solution with $\max_j \sum_i x_{ij}p_{ij} = pn/m$ can be obtained with the McNaughton's algorithm~\cite{McNaughton1959}. However, if $Q = \lceil n/m \rceil$, the things get a little complicated, as a single machine can be assigned at most $\lfloor n/m \rfloor + 1$ instances. The actual maximum CPU usage strongly depends on the value of $n \bmod m$, which we assumed is strictly greater than $0$. For example, if $n \bmod m > m/2$, then there exists an optimal schedule in which at least one machine will be assigned $\lceil n/m \rceil$ instances, none of which will be duplicated on any other machine. Thus, the optimal value of $\max_j \sum_i x_{ij}p_{ij}$ must be equal to $p\cdot\lceil n/m \rceil = p\cdot(\lfloor n/m \rfloor + 1)$. On the other hand, if $n \bmod m = 1$, then this single additional instance can be placed on each machine and the load can be evenly divided between all its instances leading to $\max_j \sum_i x_{ij}p_{ij} = p\cdot(\lfloor n/m \rfloor + 1/m)$. In general, if $Q = \lceil n/m \rceil$ and $n \bmod m > 0$, then the optimal value of $p\cdot(\lfloor n/m \rfloor + r)$ can be reached where $r = 1 / \lfloor m / (n \bmod m) \rfloor$.

The latter results require some justification. Let us observe that in the considered case ($Q = \lceil n/m \rceil$ and $n \bmod m > 0$) there always exists an optimal assignment in which each machine hosts at least $\lfloor n/m \rfloor$ single-instances applications and \emph{at most} one of the remaining $n \bmod m$. Thus, let $n' = n \bmod m$, which implies $0 < n' < m$. The exact value of $r$ can be derived by answering the following question: \emph{what is the optimal CPU usage for $m$ machines and $n'$ applications, if each machine can host only one application?} If $n'$ divides $m$, then one can split each application into exactly $m / n'$ identical instances, one per machine. Optimally, every instance would then process exactly $1 / (m / n')$ of the load $p$. In general, it may be not possible to split each application into exactly $m / n'$ identical instances. As we want to maximize the minimum number of instances, in the optimal assignment some applications will have $\lfloor m / n'\rfloor$ instances, while some (zero, if $n'$ divides $m$) will have $\lfloor m / n'\rfloor + 1$ instances. Thus, the minimal achievable fraction of load processed by a single instance will be $1 / \lfloor m / (n \bmod m) \rfloor$ of $p$, as the fewer instances, the more load each of them needs to process.

\begin{algorithm}
    \caption{Constructing the assignment that guarantees the minimal maximum vCPU usage if $p_i = p$ and $q_i = q$. For a fixed value of $m$, this algorithm works in $O(n)$.}
    \label{alg:min-maxp-pq}
    \begin{algorithmic}
        \State{$Q \gets q\cdot(Q \div q)$} \Comment{Make $Q$ a multiple of $q$}
        \State{$\mathcal{A} \gets$ a stack of all the applications}

        \For{$j \gets 1, 2, \dots, m$}
            \For{$k \gets 1, 2, \dots, \lfloor n/m\rfloor$}
                \State{$J \gets \textsc{Pop}(\mathcal{A})$}
                \State{Assign $p$ units of vCPU load of application $J$}
                \State{\quad to the $j$-th machine}
            \EndFor
        \EndFor

        \State{$r \gets 0$}
        \If{$|\mathcal{A}| > 0$} \Comment{There are still applications left}
            \If{$Q > q\cdot\lceil n/m \rceil$} \Comment{Use McNaughton's algorithm}
                \State{$r \gets p\cdot(n/m - \lfloor n/m\rfloor)$} \Comment{Available machine load}

                \State{$i \gets 1$} \Comment{Index of the considered machine}
                \State{$rr \gets r$} \Comment{Disposable load on the machine}
                \While{$i < m$}
                    \State{$J \gets \textsc{Pop}(\mathcal{A})$}
                    \State{Assign $rr$ units of CPU load of}
                    \State{\quad application $J$ to the $i$-th machine}
                    \State{$i \gets i + 1$} \Comment{Move to the next machine}
                    \State{Assign $p - rr$ units of CPU load of}
                    \State{\quad application $J$ to the $i$-th machine}
                    \State{$rr \gets r - (p - rr)$}
                \EndWhile
            \Else \Comment{Assign applications unevenly}
                \State{$r \gets p / \lfloor m / (n \bmod m) \rfloor$} \Comment{Available machine load}
                
                \State{$i \gets 1$} \Comment{Index of the considered machine}
                \State{$rp \gets 0$} \Comment{Disposable load of an application}
                \While{$i \leq m$}
                    \If{$rp = 0$}
                        \If{$\textsc{Empty}(\mathcal{A})$}
                            \State{\textbf{end while}}
                        \EndIf
                        \State{$J \gets \textsc{Pop}(\mathcal{A})$}
                        \State{$rp \gets p$}
                    \EndIf
                    \State{Assign $\min\{rp, r\}$ units of CPU load of}
                    \State{\quad application $J$ to the $i$-th machine}
                    \State{$rp \gets r - \min\{rp, r\}$} \Comment{Load left for app. $J$}
                    \State{$i \gets i + 1$}
                \EndWhile
            \EndIf
        \EndIf
    \end{algorithmic}
\end{algorithm}

\subsection{Arbitrary CPU or memory requirements}
\label{subsec:minmaxcpu-arbitrary}

If either CPU or memory requirements are arbitrary (i.e. application-dependent), then the problem becomes NP-Hard.

\begin{lemma}
Minimizing the maximum CPU usage is NP-Hard: 
\begin{enumerate}
\item for unit CPU requirements ($\forall i: p_i=1$) and arbitrary memory requirements; 
\item for unit memory requirements ($\forall i: q_i=1$) and arbitrary CPU requirements.
\end{enumerate}
\end{lemma}

\begin{proof}
We show that \textsc{3-Partition} (where we need to split a set of $3m$ values into $m$ threes in such a way that the sum of each three is the same) reduces to these special cases. Let us be given a set of $3m$ positive integers $a_1, a_2, \dots, a_{3m}$ (with their sum being a multiple of $m$). We create an instance of the balancing problem 
with $n=3m$ applications such that, depending on the case:
\begin{enumerate}
    \item $p_i = 1$ and $q_i = a_i$,
    \item $q_i = 1$ and $p_i = a_i$,
\end{enumerate}
for each $i \in \{1, 2, \dots, 3m\}$. We are also given $m$ machines such that, depending on the case:
\begin{enumerate}
    \item $Q = \frac{1}{m} \sum_{i=1}^{3m} q_i$,
    \item $Q = 3$.
\end{enumerate}
The question is, depending on the case, \emph{does there exist an assignment of instances to machines such that}:
\begin{enumerate}
    \item $\max_j \sum_i x_{ij}p_{ij} \leq 3$?,
    \item $\max_j \sum_i x_{ij}p_{ij} \leq \frac{1}{m} \sum_{i=1}^{3m} p_i$?.
\end{enumerate}
Note that any \textsc{3-Partition} instance is a yes-instance if and only if the corresponding instance of our problem is a yes-instance. As the transformation can be performed in polynomial time, the NP-Hardness follows.
\end{proof}

\section{Minimizing the number of machines used}
\label{sec:min-cpu-number}
In the bin-packing variant, 
we minimize the number of used machines
with a constraint $P$ on the maximum vCPU utilization on any machine.
In the ILP formulation, we introduce an indicator binary variable $y_j \in \{0, 1\}$
that marks a machine that is used (that is assigned some load):
\begin{equation}
    \sum_i x_{ij}p_{ij} \leq P y_j \text{.}
\end{equation}
The number of machines (hence, the number of $y_j$ variables) is upper bounded by $\sum_{i=1}^n \lceil p_i/P\rceil$, corresponding to an allocation in which each machine is assigned at most one instance of some application, with the vCPU capacity of $P$, until all the CPU requirements are met.
The goal is thus: 
\begin{equation*}
    \min \sum_j y_j \text{.}
\end{equation*}
We analyze special cases of: (1) unit; (2) common; and (3) arbitrary CPU or memory requirements.

\subsection{Unit CPU and memory requirements}
This special case is simple. Consider a machine for which the values of $P$ and $Q$ are given. Note that both the $P$ and $Q$ values are assumed to be integers. 
In a single-instanced model, 
exactly $\min\{P, Q\}$ applications are allocated on each machine $j$. 
Surprisingly, having multiple instances does not decrease the number of used machines, as the following lemma shows.

\begin{lemma}
\label{lemma:no-of-machines-1-1-no-load}
If $p_i = q_i = 1$ for each $i$, then there exists an optimal assignment of instances to machines where all applications have exactly one instance.
\end{lemma}

\begin{proof}
Let us consider any optimal assignment in which an application has more than one instance. For each machine $j,$ let us denote by $\mathcal{I}_j$ the set of all applications $i$ placed on that machine ($x_{ij} = 1$), but only with partial load ($p_{ij} < 1$). Thus, there must exist at least two machines, $j$ and $k$, for which $|\mathcal{I}_j| > 0$ and $|\mathcal{I}_k| > 0$. In consequence, the set $\mathcal{I} = \bigcup_j \mathcal{I}_j$ is not empty. Notice that the total CPU usage of the applications in $\mathcal{I}$ is $|\mathcal{I}|$, and that the total memory usage of these applications is $\sum_j |\mathcal{I}_j| \geq 2|\mathcal{I}|$. 

We now reassign applications to machines in such a way that the number of used machines does not increase, but all the applications have exactly one instance. Let all the applications from $\mathcal{I}$ be removed from all the machines. Now, consider any machine $j$ that was affected by this operation. If one assigns to this machine as many complete ($p_{ij} = 1$) applications $i$ from the $\mathcal{I}$ set as possible, the total CPU load on machine $j$ will be not lower than the initial load. If this is so for all the machines, the lemma follows. Notice that the statement might be not true only if the number of applications left to be assigned is lower than the capacity of the selected machine. However, in such a case the lemma also follows, as the limit of $P$ is not achieved. 
\end{proof}

Lemma~\ref{lemma:no-of-machines-1-1-no-load} leads us to the following greedy approach.
Determine the critical capacity $\min\{P, Q\}$ of a machine.
Then, until there are no applications left to be assigned, place
as many as possible complete ($p_{ij} = 1$) instances of applications on a new machine.

The above argument generalizes to any case in which $p\mid P$ and $q\mid Q$. Indeed, if $q\mid Q$ then the memory requirements can be scaled to $q = 1$ with $Q$ changed to $Q \div q$ (same argument holds for $p$).

\subsection{Common CPU and memory requirements}

As seen in the previous section, when $p\mid P$ and $q\mid Q$, there exists an optimal assignment of applications to machines such that no application uses more than one instance. 
If $q\nmid Q$, then without the loss of generality, we can reduce $Q$ to $q\cdot(Q \div q)$ (as no application would fit into the remaining capacity $Q \mod q$ anyhow). 
Notice that $q\mid q\cdot(Q \div q)$. In other words, we can always assume---without the loss of generality---that $q\mid Q$ and, equivalently, that $q = 1$. 

However, if $p\nmid P$, then Lemma~\ref{lemma:no-of-machines-1-1-no-load} is not true anymore, as the following example shows.

\begin{example}
Assume $Q = 2$ and $P = 3$, and that we are given three applications, $\{\ja, \jb, \jc\}$, with $q_i = 1$ and $p_i = 2$ (\mbox{$p\nmid P$}). 
In the single-instanced model, the optimal solution allocates each application to a separate machine, with $m=3$. 
A multi-instanced model allows us to split one of the applications as in Fig.~\ref{fig:ex-assignment-p-nmid-P}, and in consequence use just two machines.
\end{example}

\begin{figure}
    \centering
    \includegraphics[scale=0.85]{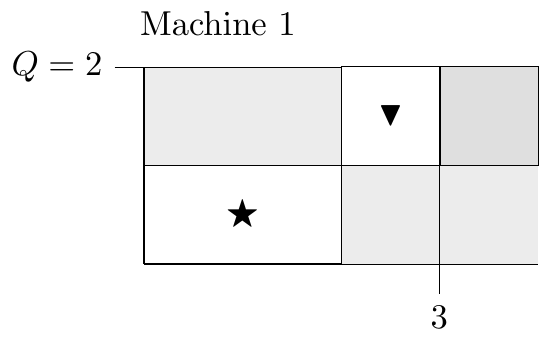}
    \includegraphics[scale=0.85]{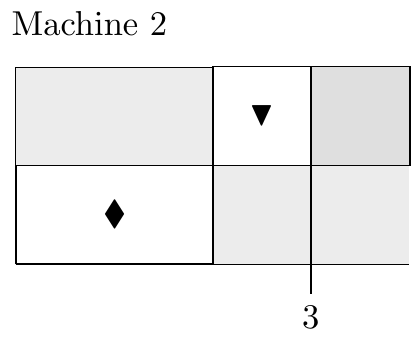}
    \caption{When $\jc$ has more instances, the allocation uses just two machines (while a single-instanced allocation uses three).}
    \label{fig:ex-assignment-p-nmid-P}
\end{figure}

Applications using multiple instances can reduce the number of used machines by the factor of almost two, as proved by the following lemma.

\begin{lemma}
\label{lemma:ratio2}
Let $m$ be the optimal number of machines for the single-instanced model; and let $m'$ be the optimal number of machines for the multi-instanced model. 
(1) For any instance, $m/m' < 2$; (2) for any given real $\varepsilon > 0$, there exists an instance such that $2(1 - \varepsilon) < m/m' < 2$.
\end{lemma}

\begin{proof}
We first show by contradiction that for any instance $m/m' < 2$. Notice that $m \geq m'$ and that $m = 1$ if and only if $m' = 1$. Thus, assume that $m' > 1$, $m \geq 2m'$, and that the values of $Q$, $P$, $q$ and $p$ are arbitrary. Consider the single-instanced model (which requires $q \leq Q$ and $p \leq P$). Let $L$ be the maximum CPU load among all the machines. As applications have common requirements, in any optimal solution, applications can be relocated in such a way that on the first $m - 1$ machines the load is equal to $L$, and on the $m$-th machine the load $L'$ is positive, yet not greater than $L$.  
Now, consider the multi-instanced model. As $m \geq 2m'$ and $m' > 1$, there must exist at least one machine for which the total CPU load $H$ meets the condition $H \geq L + L'$. This is so, as the average load on $m'$ machines is equal to:
\begin{equation*}
\begin{split}
\frac{(m-1)L + L'}{m'} & \geq \frac{(m-1)L + L'}{m/2} = \\
& = \left(2 - \frac2m\right)L + \frac2mL' =\\
& = L + \left(1 - \frac2m\right)L + \frac2mL' \geq \\
& \geq L + \left(1 - \frac2m\right)L' + \frac2mL' = \\
& = L + L'.
\end{split}
\end{equation*}

However, by the above,
the number of applications assigned to this machine is at least $(L+L')/p$. Thus, $Q \div q \geq (L+L')/p$ and $P \geq L + L'$. Consequently, in the single-instanced model, machines $m - 1$ and $m$ could be merged --- which means $m$ is not optimal, leading to a contradiction.

Now, let us show that, for any $\varepsilon > 0$, there exists an instance such that $2(1 - \varepsilon) < m/m'$. Let $n > 2$ be the number of different applications, and let $q_i = 1$ for each $i \in \{1, 2, \dots, n\}$. Let also $Q = 3$, $p = n$ and $P = 2n - 1$. If we consider a single-instanced model, the optimal number of machines, $m$, is equal to $n$. On the other hand, the total load processed by $n$ applications is $n^2$ and each of the machines is capable of processing the load of size $2n - 1$. As a consequence,
the lower bound on the number of machines $m'$ for the multi-instanced model is $\lceil n^2/(2n-1)\rceil$. An optimal assignment using exactly $\lceil n^2/(2n-1)\rceil$ machines, where at most three different applications are assigned to each machine, can be found by assigning jobs to machines greedily, one machine after another. Thus,
\begin{equation*}
\begin{split}
\frac{m}{m'} & = \frac{n}{\lceil n^2/(2n-1)\rceil} \\
             & > \frac{n}{n^2/(2n-1) + 1} = \frac{2n^2 - n}{n^2 + 2n - 1} \xrightarrow[n \to \infty]{} 2_{-}.\qedhere
\end{split}
\end{equation*}
\end{proof}

In order to find an optimal assignment in polynomial time, we observe that:
\begin{enumerate}
\item on any machine, we can allocate $\min\{P\div p, Q \div q\}$ whole applications;
\item it is suboptimal to allocate many instances of applications to a single machine if a single instance of the application could be assigned instead (cf. Lemma.~\ref{lemma:no-of-machines-1-1-no-load});
\item if the optimal load ($\min \max_j \sum_i x_{ij}p_{ij}$) is known, we can compare it to $P$ to see whether additional machines are necessary.
\end{enumerate}
Thus, the general idea of the algorithm is as follows (Alg.~\ref{alg:min-no-machines-pq}). Given the lower $m_L$ and the upper $m_U$ bounds on the possible number $m$ of machines, we test---based on binary search procedure and Sect.~\ref{subsec:minmaxcpu-qp}---whether the optimal maximum CPU load exceeds $P$. Depending on the answer, we limit the range of $m$. 
Note that Alg.~\ref{alg:min-no-machines-pq} works for the case of $p\mid P$, too. Also note that, given the optimal number of machines $m$, one can assign applications to machines based on the Alg.~\ref{alg:min-maxp-pq}.

\begin{algorithm}
\caption{Finding the minimal number of machines if $p_i = p$ and $q_i = q$}
\label{alg:min-no-machines-pq}
\begin{algorithmic}
    \State {$m_L \gets \lceil\max\{n/(Q \div q), n\cdot p/P\}\rceil$} \Comment{Lower bound on the number of machines}
    \State {$m_U \gets n\cdot\lceil p/P \rceil$}

    \While{$m_U > m_L$}
        \State{$m \gets (m_L + m_U) \div 2$}
        \If{$n \bmod m = 0$ or $Q \div q > \lceil n/m \rceil$}
            \State{$P' \gets p\cdot n/m$} 
        \Else
            \State{$P' \gets p\cdot(\lfloor n/m \rfloor + 1/\lfloor m/(n \bmod m)\rfloor)$}
        \EndIf
        \If{$P' > P$}
            \State{$m_L \gets m + 1$}
        \Else
            \State{$m_U \gets m$}
        \EndIf
    \EndWhile

    \State{\Return $m_L$}
\end{algorithmic}
\end{algorithm}

\subsection{Arbitrary CPU or memory requirements}

Analogically to the load balancing problem (Section~\ref{subsec:minmaxcpu-arbitrary}), 
if either CPU or memory requirements are arbitrary (i.e. application-dependent), then the problem becomes NP-Hard. 

\begin{lemma}
\label{lemma:min-no-of-machines-NP}
Minimizing the number of machines used is NP-Hard:
\begin{enumerate}
    \item for unit CPU requirements ($\forall i: p_i=1$) and arbitrary memory requirements; 
    \item for unit memory requirements ($\forall i: q_i=1$) and arbitrary CPU requirements.
    \end{enumerate}
\end{lemma}

\begin{proof}
We show that \textsc{3-Partition} reduces to these special cases. Let us be given a set of $3m$ positive integers $a_1, a_2, \dots, a_{3m}$ (with their sum being a multiple of $m$). We create an instance of our problem with  $3m$ applications such that, depending on the case:
\begin{enumerate}
    \item $p_i = 1$ and $q_i = a_i$,
    \item $q_i = 1$ and $p_i = a_i$,
\end{enumerate}
for each $i \in \{1, 2, \dots, 3m\}$. We are also given $m$ machines such that, depending on the case:
\begin{enumerate}
    \item $Q = \frac{1}{m} \sum_{i=1}^{3m} q_i$ and $P = 3$,
    \item $P = \frac{1}{m} \sum_{i=1}^{3m} p_i$ and $Q = 3$.
\end{enumerate}
The question is: \emph{does there exist an assignments of instances to machines such that $\sum_j y_j \leq m$?} Note that any \textsc{3-Partition} instance is a yes-instance if and only if the corresponding instance of our problem is a yes-instance. As the transformation can be performed in polynomial time, the NP-Hardness follows.
\end{proof}

\section{Heuristics}
\label{sec:heur}

As with arbitrary requirements both the balancing and the bin-packing problems are NP-hard, in this section we propose a number of heuristics. We focus on bin-packing (minimizing the number of used machines), partly due to space constraints, partly because this problem is perhaps more applicable of the two (as the machine's CPU capacity should not be exceeded in the steady state).
We start with baselines: heuristics packing single-instanced applications. Then, we extend them for the multi-instanced model.

A cloud cluster usually consists of large groups of similar machines. For each of these machines, the CPU capacity and memory size is known in advance. Moreover, usage limits can be set on these resources. For example, on a single machine, one may not want to exceed 95\% of memory capacity and 80\% of the total CPU capacity. These thresholds directly translate to values of $P$ and $Q$.

Before we introduce our heuristics, we discuss the limits on the values of $p_i$ and $q_i$. Let us consider an $i$-th application. It must hold that $q_i \leq Q$ as otherwise an instance of this application could not be placed on any machine. However, our baselines for single-instanced model force us to assume that for each $i$-th application $p_i \leq P$. In general, one could assume that if there exists an application for which $p_i > P$, then this application could be assigned --- in advance --- to $p_i \div P$ dedicated machines. Then, the remaining $p_i \bmod P$ units of load could be assigned to some other machine based on the considered heuristic. This approach may not lead, in general, to optimal solutions, as Example~\ref{example:piload} shows. Thus, we cannot introduce it in our baselines.

\begin{example}
\label{example:piload}
Let us consider a multi-machine environment where $P = 5$, $Q = 3$, and three different applications, $\{\ja, \jb, \jc\}$, such that $p_\ja = p_\jb = q_\ja = q_\jb = 2$, $p_\jc = 6$ and $q_\jc = 1$. If $P$ units of application's $\jc$ load will be assigned to a dedicated machine, then the total number of machines required to process all the load will be equal to $3$ (Fig.~\ref{fig:ex-assignment-piload}(a)) while the optimal number is equal to $2$ (Fig.~\ref{fig:ex-assignment-piload}(b)).

\begin{figure*}[t]
    \begin{subfigure}[b]{1.2\columnwidth}
        \centering
        \includegraphics[scale=0.55]{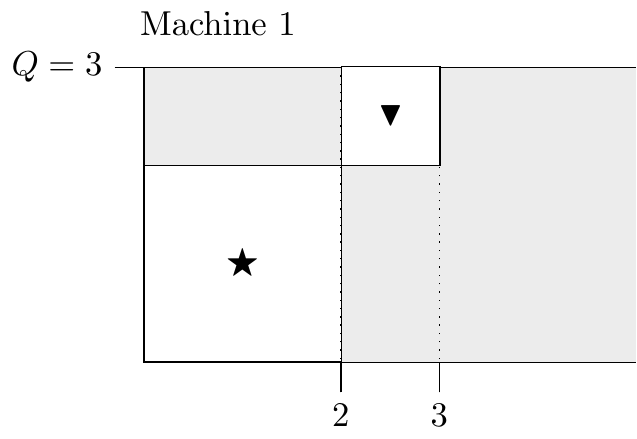}
        \includegraphics[scale=0.55]{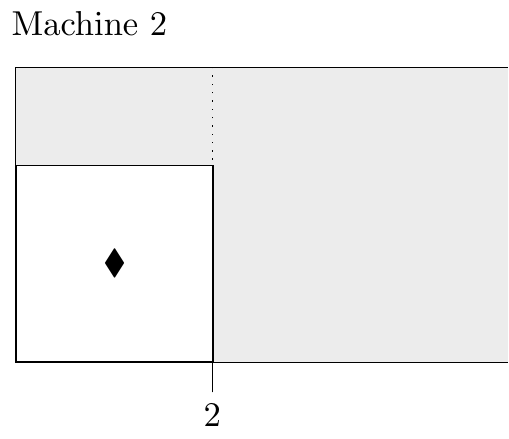}
        \includegraphics[scale=0.55]{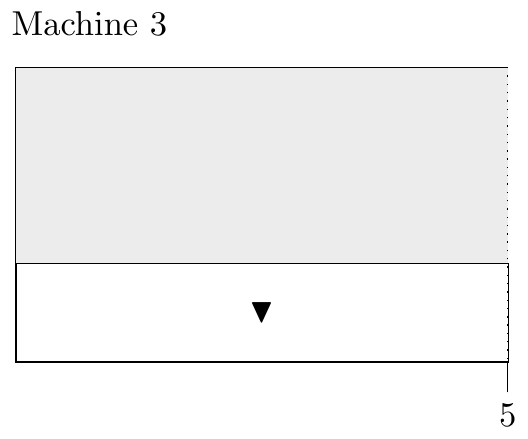}
        \caption{The $\jc$ application is forced to use a dedicated Machine 3}
    \end{subfigure}
    \begin{subfigure}[b]{0.9\columnwidth}
        \centering
        \includegraphics[scale=0.55]{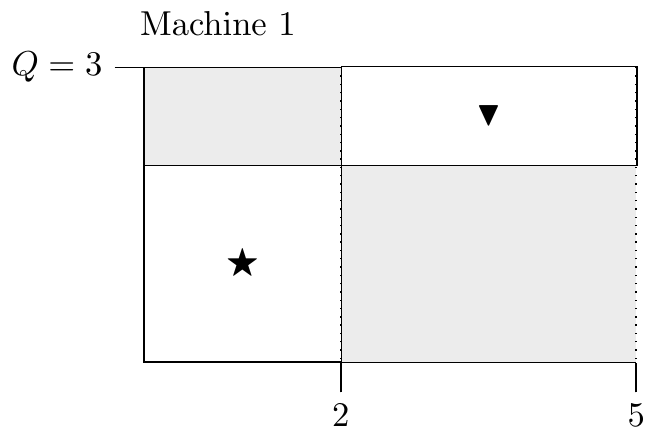}
        \includegraphics[scale=0.55]{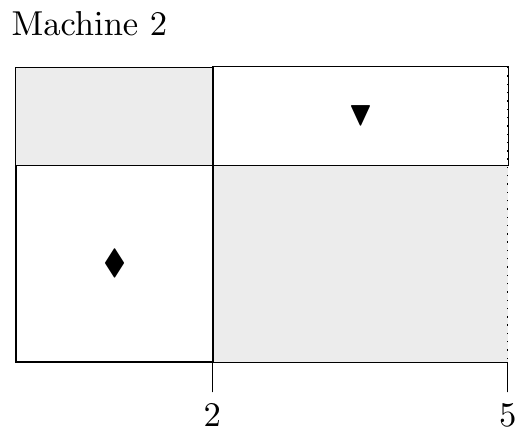}
        \caption{Optimal assignment}
    \end{subfigure}
    \caption{Example assignment of three applications to two machines. On the left, we assume that an application for which $p_i > P$ is assigned a dedicated machine. On the right, we present an optimal assignment.}
    \label{fig:ex-assignment-piload}
\end{figure*}
\end{example}

\subsection{Baselines for the single-instanced model}

As the baselines, we use standard list heuristics parametrized on two levels: (1) for a specific application, how to choose the machine; and (2) the order in which the applications are processed. We consider the following rules:
\begin{itemize}
    \item \textsc{First-Fit}. Find the first machine on which the application $i$ fits (the machine's free CPU is at least $p_i$; and the machine's free memory is at least $q_i$).
    \item \textsc{Next-Fit}. Find the next machine on which the application $i$ fits (in a cyclic way: after reaching the last opened machine, we start with the first one).
    \item \textsc{Worst-Fit}. Find the machine for which, after the application $i$ is placed, the remaining free memory is largest.
\end{itemize}
Each of these rules opens a new machine if the application does not fit on any machine. Also note that we use worst-fit, rather than best-fit, in order to leave as much free memory on a machine as possible for the coming applications.

We consider the following application orders:
\begin{itemize}
    \item \textsc{Mem-Increasing}/\textsc{Mem-Decreasing}. Choose the application with the largest/lowest memory requirement.
    \item \textsc{CPU-Increasing}/\textsc{CPU-Decreasing}. Choose the application with the largest/lowest CPU requirement.
    \item \textsc{Ratio-Increasing}/\textsc{Ratio-Decreasing}. Choose the application with the largest/lowest value of $p_i/q_i$, i.e., the ratio of CPU and memory requirements.
    \item \textsc{Random}. Choose applications in a random order.
\end{itemize}

\subsection{Heuristics for the multi-instanced model}

Our heuristics are based on the binary search (Algorithm~\ref{alg:min-no-machines-pq}). 
We explore the range of possible numbers of machines $[m_L, m_U]$. For a tested $m$, we apply one of the heuristics below and check whether the assignment does not not exceed $P$ on any used machine. We consider the following heuristics:

\begin{itemize}
    \item \textsc{CPU-Oriented}. This heuristic is based on the following assumption: \emph{we want to add instances to applications that have the highest CPU requirements}. The whole applications are placed on machines based on the \textsc{First-Fit}/\textsc{Next-Fit}/\textsc{Worst-Fit} rule, starting from the applications with the lowest CPU requirements. It may happen that at some point none of the remaining applications will fit any of the machines without adding instances. Then, one application after another, starting from the one with the largest memory requirements, we greedily fill the machines with as large instances as possible. We do it based on the \textsc{First-Fit}/\textsc{Next-Fit}/\textsc{Worst-Fit} rule, and taking into account the limits of $P$ and $Q$.
    \item \textsc{Mem-Oriented}. This heuristic is based on the following assumption: \emph{we want to add instances to applications that have the lowest memory requirements}. The whole applications are greedily placed on machines based on the \textsc{First-Fit}/\textsc{Next-Fit}/\textsc{Worst-Fit} rule, starting from the applications with the largest memory requirements. It may happen that at some point none of the remaining applications will fit any of the machines without adding instances. Then, one application after another, we greedily fill the machines with as large instances as possible. We do it based on the same rule, and taking into account the limits of $P$ and $Q$.
\end{itemize}

\section{Experiments}
\label{sec:experiments}

Although the theoretical results of Lemma~\ref{lemma:ratio2} are promising, they do not take into account the peculiarities of load processed in clouds. For this reason, we evaluate our heuristics using instances generated from the Azure Public Dataset V2 (VM Trace).

\subsection{Data preprocessing}

The Azure Public Dataset V2 (VM Trace) provides information about Azure VM workload collected over 30 consecutive days in 2019.
The requirements of each VM are bucketed by their memory usage (0-2 GB, 2-4 GB, 4-8 GB, 8-32 GB, 32-64 GB or more than 64 GB) and core count (0-2, 2-4, 4-8, 8-12, 12-24, and more than 24). 
The VMs in the trace are grouped into \emph{deployments} --- sets of virtual machines deployed and managed together by a single client.
We map each deployment to a separate application. 
We also filter out 513 out of 16,977 deployments which use buckets with indefinite upper bounds.

We map $q_i$ to 
the maximum
upper end-point of the memory buckets assigned to all VMs from the deployment. 
In 72\% of the deployments, all VMs are in the same memory buckets; and in  further 14\%, they belong to exactly two buckets.

To derive $p_i$, we sum the average (fractional) usage of the vCPU cores over VMs in the deployment. As the trace defines only the range of vCPUs used by a VM, we use the upper end of the VM's CPU bucket. For example, if a VM is assigned a bucket of 8-12 cores, and its average CPU usage is $0.6$, then we map that to $12\cdot 0.6 = 7.2$ virtual cores. 

We thus use the application's memory \emph{limits} (upper bounds) for memory requirements, but the application's CPU actual \emph{usage} for the load.
This is so as the CPU --- a compressible resource --- is easier to vertically-scale without much disruption to the running VM, so the CPU limit should be closer to the (perhaps high percentile of) CPU usage.

In our experiments,
we use machines with $P=32$ virtual cores and varying capacity $Q$ of RAM: 64, 96, 128, 256, 512 and 1024 GB (we use $Q$ as a parameter of the experiment).
In an appendix~\cite{appendix}, we present results for additional configurations.
We start with 64 GB, as $99.6\%$ of VMs in the trace have at most 64GB of RAM. We stop at 1024 GB, because, as we later show, memory ceases to be a critical resource from roughly this value (thus, a rational cloud provider would not have machines of this size).
We also restrict instances to applications with the total vCPU requirement of at least 1 (VMs that use less than 1 vCPUs share it with other VMs which impacts the quality of service) and at most 32 (otherwise, in configurations with 64 GB  machines, some machines would host just a single VM, regardless of the algorithm used).

After this mapping, in the 16,464 applications the median $p_i$ is $5.0$ and the mean is $8.4$; the median $q_i$ is 8~GB and the mean $q_i$ is 14~GB. Thus, a 32-core machine accepts roughly 6 ``median'' applications when considering only the CPU requirements; and between 6 (for 64 GB of RAM) and 128 (for 1024 GB of RAM) ``median'' applications when considering only the memory requirements.

We generate 50 instances. Each instance has 100 applications selected randomly  from the base~set (without replacement). We simulate each heuristic on each instance and each machine configuration. 
We measure the number of used machines which directly corresponds with the average utilization (the lower the number of machines, the higher the utilization). 
To meaningfully compare results between different instances that can have different loads, we normalize the number of machines by the classic lower bound of the average load (in our CPU/memory case, the bound is extended by the average memory requirement): $\max \left( \lceil \sum_i p_i / P \rceil, \lceil \sum_i q_i / Q \rceil \right)$.
Figure~\ref{fig:results} shows the results.

\begin{figure*}[p]
	\centering
	\subfloat[64 GB]{{\includegraphics[width=0.5\textwidth]{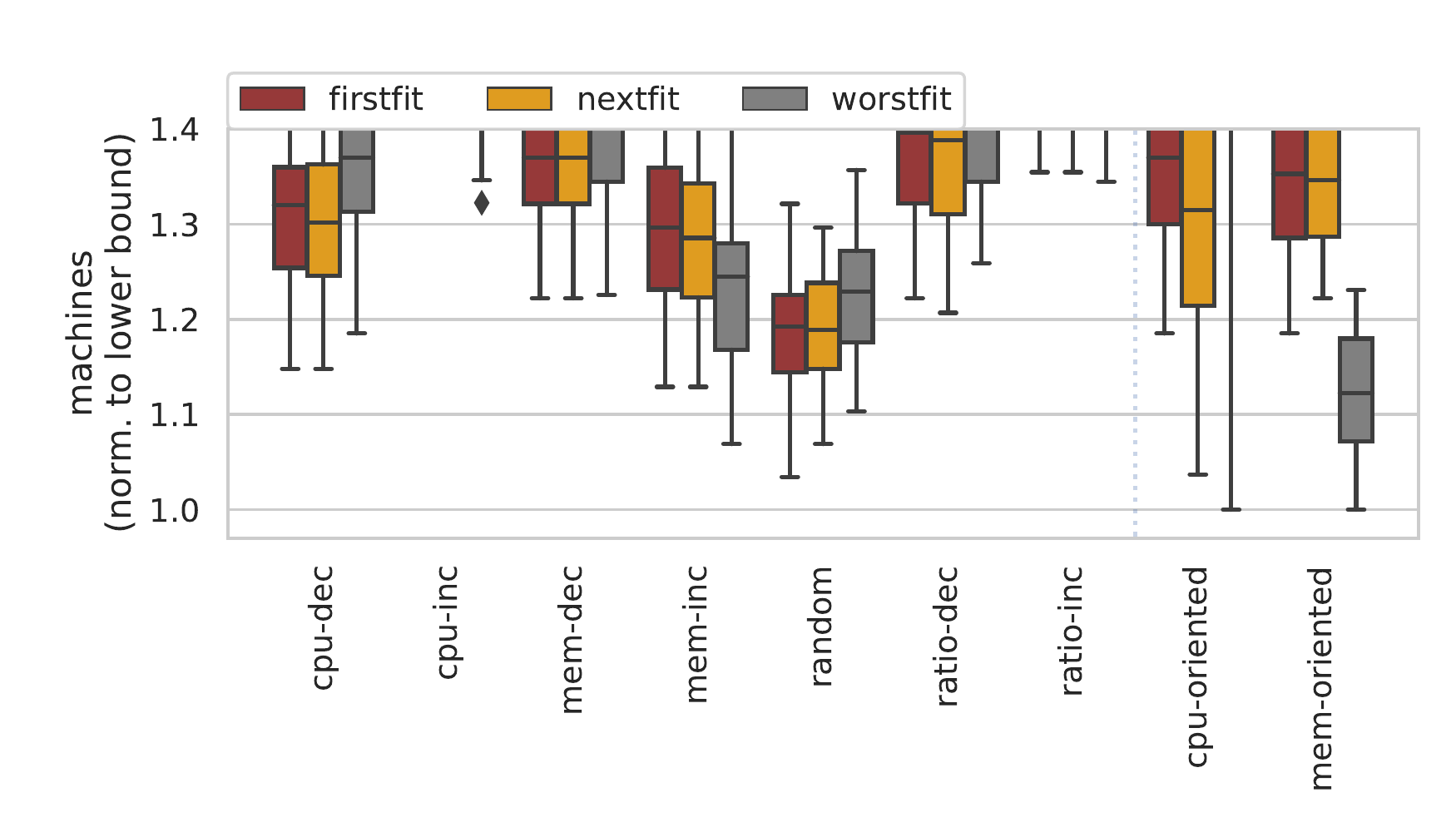}}}%
    \subfloat[96 GB]{{\includegraphics[width=0.5\textwidth]{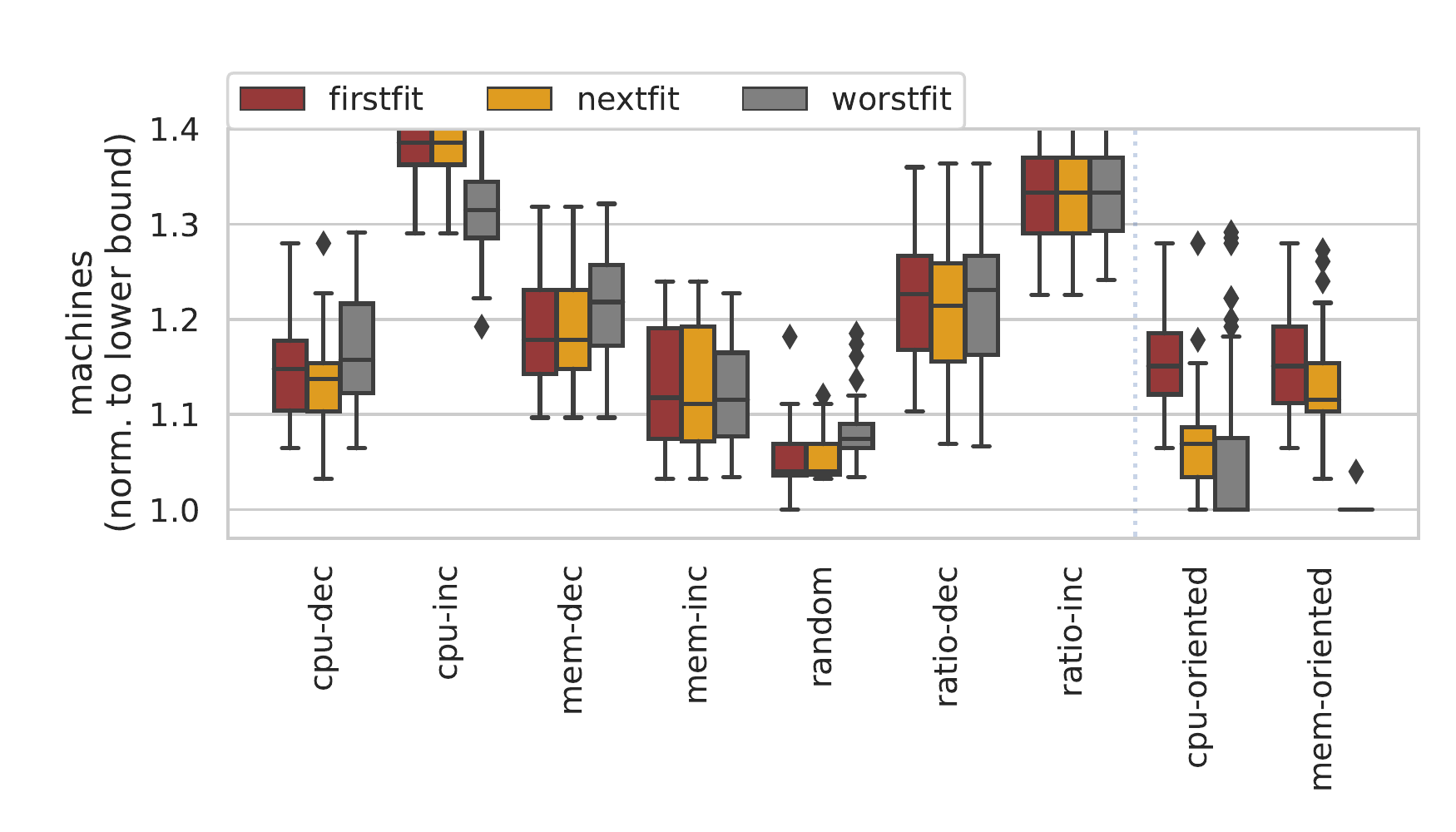}}}%
    \\
    \subfloat[128 GB]{{\includegraphics[width=0.5\textwidth]{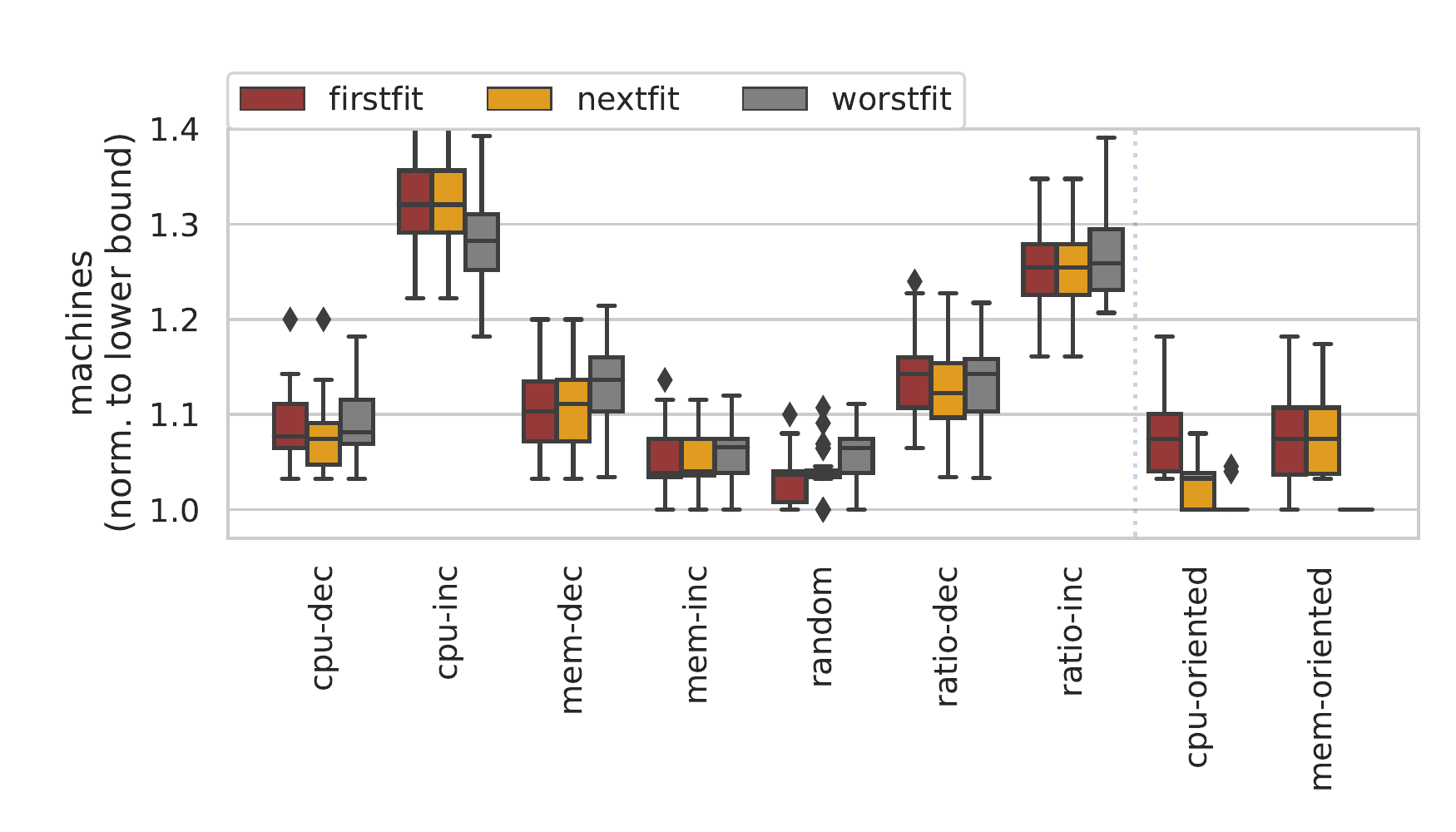}}}%
    \subfloat[256 GB]{{\includegraphics[width=0.5\textwidth]{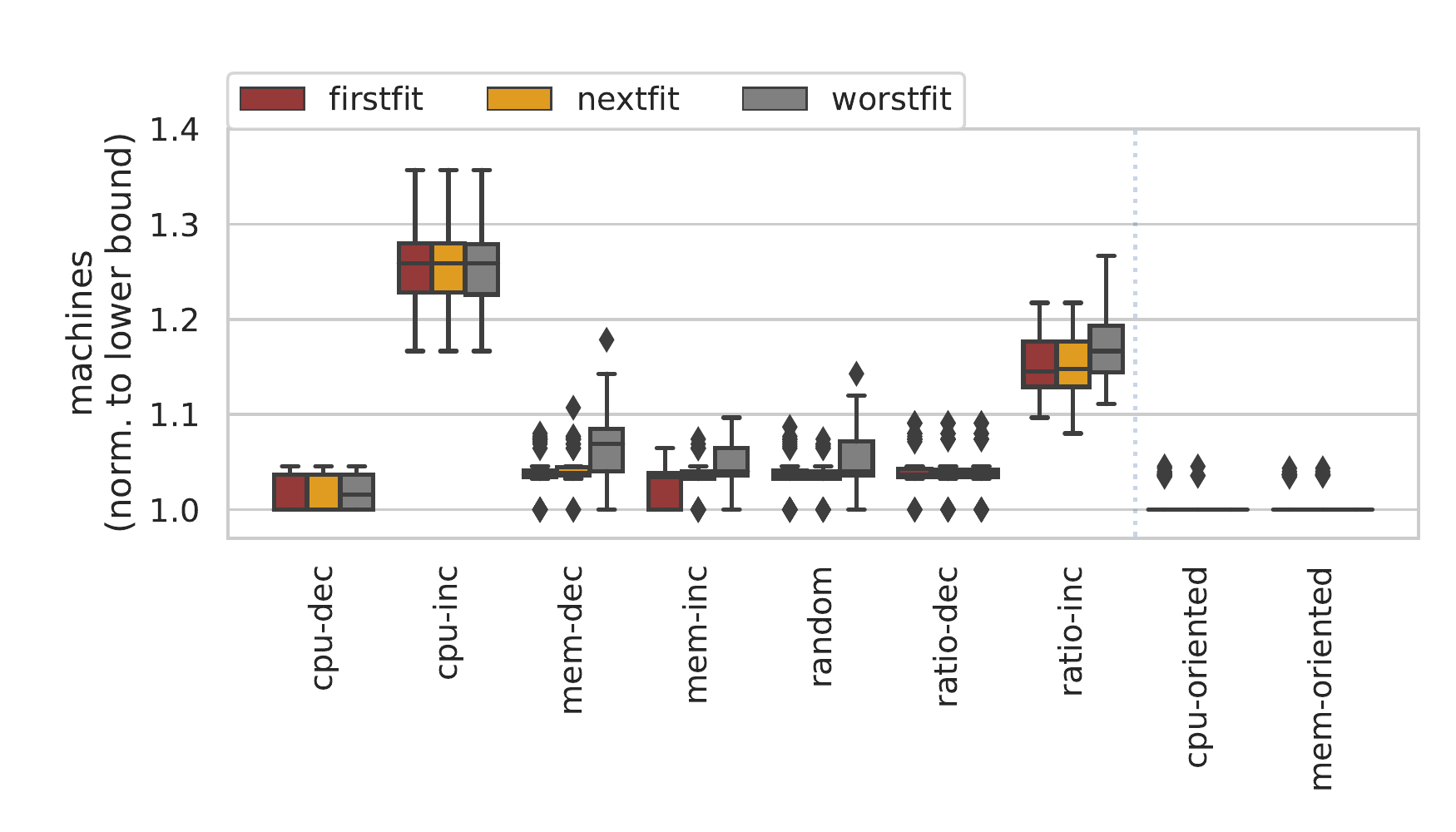}}}%
    \\
    \subfloat[512 GB]{{\includegraphics[width=0.5\textwidth]{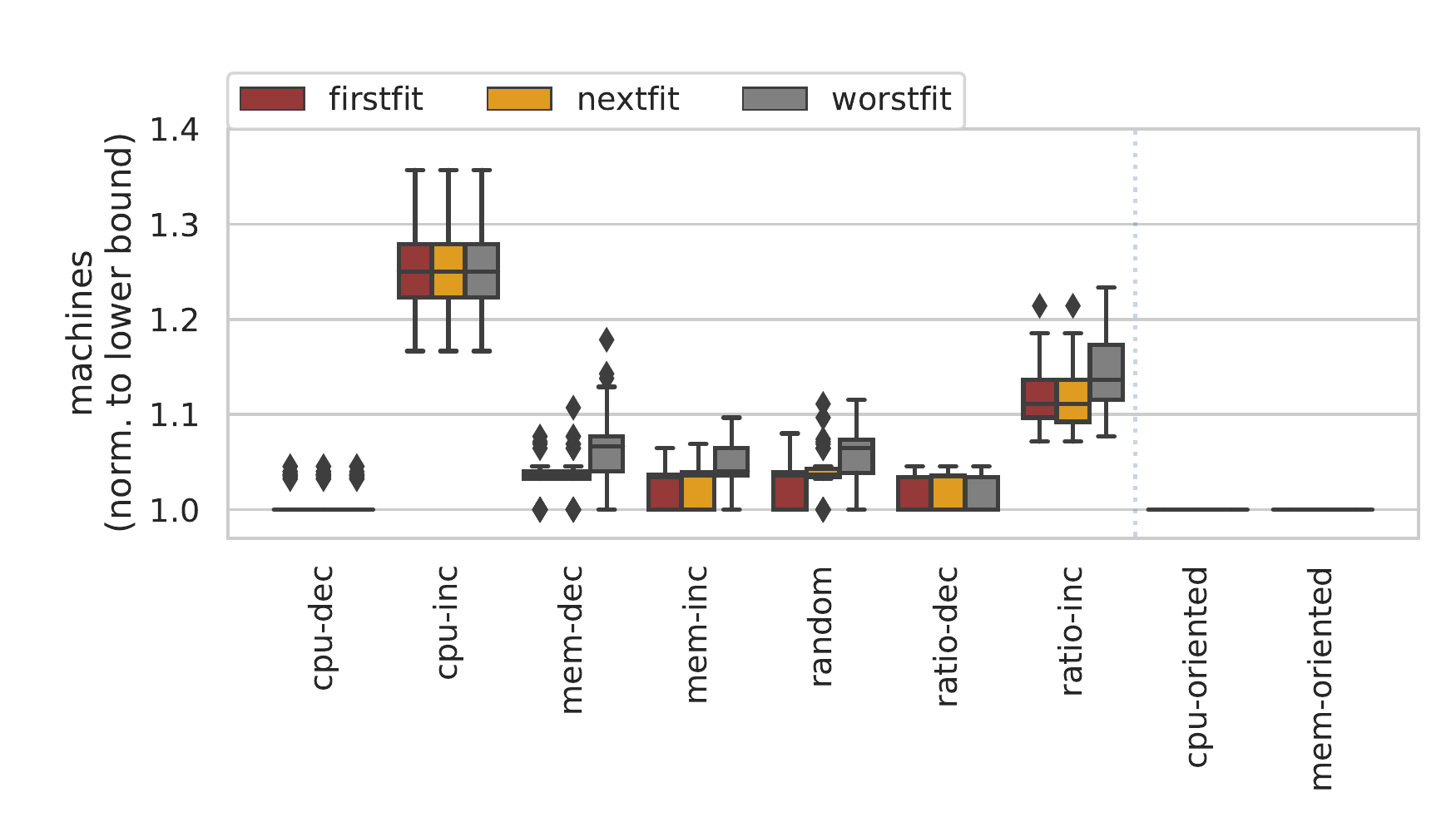}}}%
    \subfloat[1024 GB]{{\includegraphics[width=0.5\textwidth]{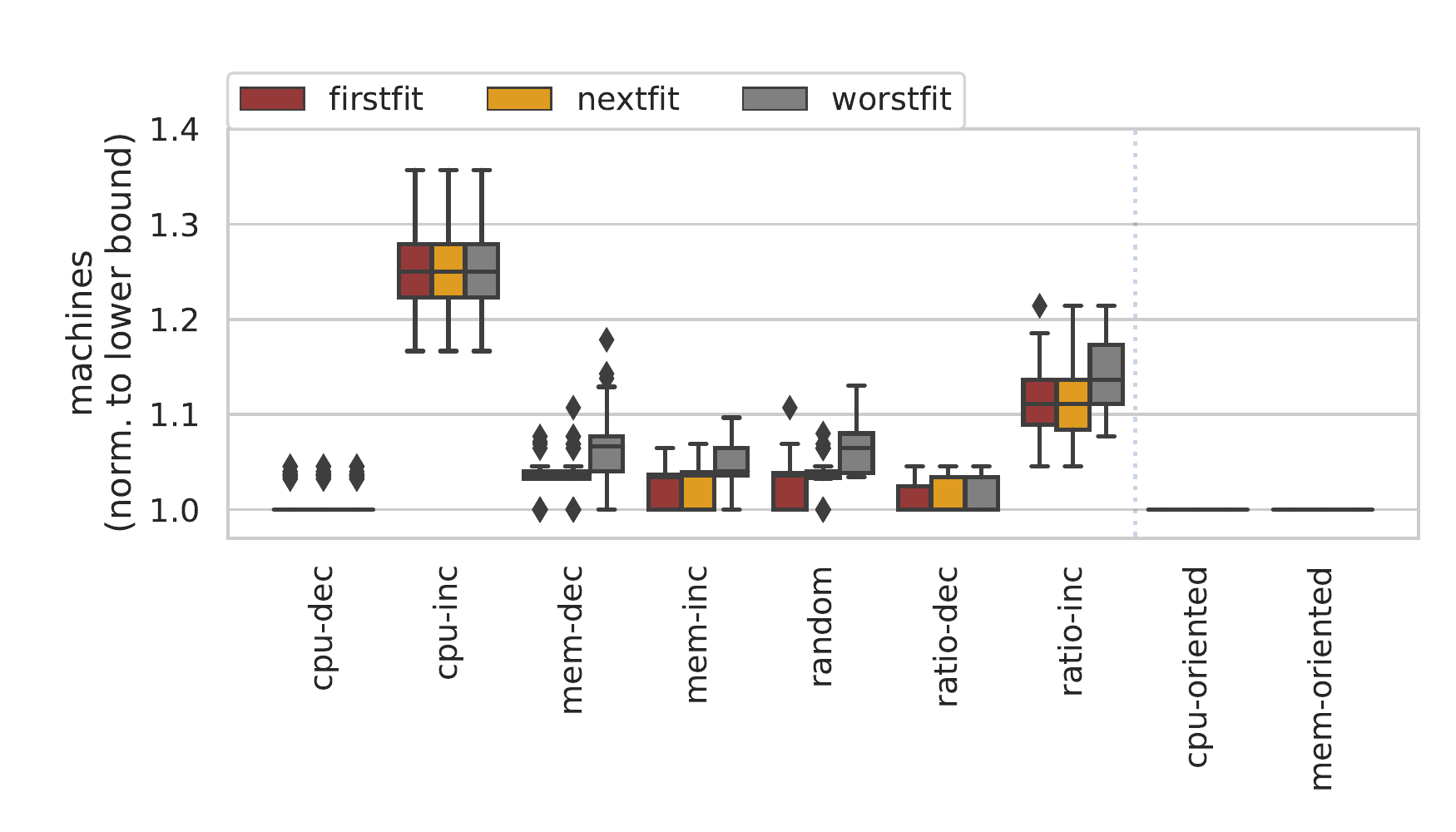}}}%
	\caption{Comparison of performance of the algorithms for machines with 32 vCPUs and varying amount of memory. Each boxplot shows statistics over 50 individual experiments.
    The line inside the box corresponds to the median and the height of the box indicates the first and the third quartile. The whiskers extend to the most extreme data point within 1.5 $\times$ IQR (inter-quartile range).}
	\label{fig:results}
\end{figure*}

\subsection{Results}

With enough memory (512-1024 GB, see Fig.~\ref{fig:results}(ef)), our heuristics generate results equal to the lower bounds. 
The reason is that for large amounts of memory, memory is no longer a scarce commodity, thus the cost of maintaining multiple instances is negligible. Thus, one can easily obtain the lower-bound of $\lceil \sum_i p_i / P \rceil$.
This holds regardless of the algorithm used, which shows that simply having multiple instances is more important than how precisely we create them. 
On the other hand, baseline heuristics perform similarly, as the machines can be almost fully utilized when no application is split.

When the amount of memory is low (64 GB, see Fig.~\ref{fig:results}(a)), adding instances is costly. Moreover, the number of instances that may be maintained by a single machine is small. Thus, algorithm performance strongly depends on the choice of the applications that have multiple instances. In Fig.~\ref{fig:results}(a), both single-instanced and mutli-instanced heuristics perform poorly (at least 10\% worse) compared to an optimistic lower-bound. However, the differences between methods are notable. The \textsc{Mem-Oriented}+\textsc{Worst-Fit} rule provides the best results (roughly  10\% over the lower bound).

In the intermediate configurations (96-256 GB, see Fig.~\ref{fig:results}(bcd)), multi-instanced heuristics show a significant improvement over single-instanced approaches, often leading to solutions achieving the lower bound. However, it still matters which applications have multiple instances: while in the case of 96 GB or RAM, \textsc{Mem-Oriented}+\textsc{Worst-Fit} is optimal, other heuristics are not --- and the differences diminish with the increased memory capacity.

\section{Related Work}
Below, we review related work in combinatorial optimization and in VM placement. We refer to \cite{10.1145/2983575} for a survey.

The divisible load~\cite{bharadwaj2003divisible} scheduling model is related to our approach: the difference is that the divisible load does not consider the memory requirements.

When single-instanced applications of size $q_i$ need to be placed on infinitely-efficient machines of capacity $Q$ in order to minimize the number of machines used, our problem reduces to \textsc{Bin-packing}.
Some heuristics solving optimally more than 95\% of analyzed instances of the \textsc{Bin-packing} problem are known \cite{FLESZAR2002821}.
The key difference between our approach and the multi-dimensional bin-packing~\cite{christensen2017approximation} is that in multi-dimensional bin-packing all requirements are fixed.

A related problem is also the one-dimensional \textsc{Fractional Bin-packing}: objects can be split across multiple bins (which is a linear programming relaxation of the \textsc{Bin-packing} problem). Variants in which the share of each object assigned to a single bin must be the same \cite{CASTROSILVA20194}, or in which packing together two or more items make them use less resources than the sum of their individual requirements \cite{GRANGE2018331} are considered.

We consider the more general case of two-dimensional bin-packing in the multi-instanced model for which one dimension is always constant while the other one can change. Thus, we now focus on selected results that are related to these characteristics.
In \cite{Morihara1983173} bin-packing problem and the multiprocessor scheduling problem are connected: they minimize the number of workers or days required to produce certain amounts of goods.
\cite{Kamali2015} analyzes how virtual machines with dynamic workload can be managed when the amount of resources required by them changes in time.
They extend the typical load balancing with live migration to keep all the virtual machines in a limited number of active nodes.
\cite{6565979} also studies VM placement with live migrations.

VM placement with bin-packing (and its variants) are considered
in \cite{IPDPS7516028,7778998,9248589,9029147,Fatima2019,6569802}.
Specifically, \cite{10.1145/3410220.3456278} shows improved approximation factors when load prediction is available.
\cite{moges2019energy} analyzes heuristics for bin-packing taking into account the power-efficiency of the host.
\cite{7792170} analyzes VM placement with memory sharing (e.g. common libraries). 

\cite{6216364,Chung2006,Epstein2009,10.1007/978-3-540-77918-6_19,7116894} are closest to our results as they map to special cases of our model. 
\cite{6216364,Chung2006} assume $q_i=1$ and find any feasible assignment. \cite{Epstein2009} shows a (3/2)-approximation algorithm for $q_i=1$ and $Q=2$; and (7/5)-approximation for an arbitrary $Q$; \cite{10.1007/978-3-540-77918-6_19} shows a PTAS for each of these cases.
\cite{7116894} analyzes $q_i=1$ and arbitrary $p_i$ on perhaps-failing machines (with high probability the assignment must fulfill the demand of each application).
In contrast, our theoretical results solve optimally in polynomial time a special case with $q_i=q$ and $p_i=p$, i.e., equal memory and processing requirements.

\section{Conclusions}
\label{sec:conc}

We study a two-dimensional resource management problem with applications having multiple instances.
While instances of an application have the same memory requirements, the CPU load can be freely balanced between them.
From systems perspective, this approach integrates the scheduler, the autoscaler and the load balancer.

We present a number of theoretical results. We consider two related objectives: (1) minimization of the maximum load processed by a single machine; (2) minimization of the number of machines used. We demonstrate that both are NP-Hard in general, even when one of the dimensions of the problem is unit-sized. 
We also show polynomial algorithms that solve special cases with applications having equal requirements.
We also provide strong theoretical motivation for having multiple instances: when bin-packing, replication may reduce the number of machines by a tight factor of $2-\varepsilon$.

For the general case of the bin-packing problem, we propose heuristics. We simulate them on instances derived from the Azure Public Dataset. In the intermediate cases of 96-256 GB of RAM, compared with various single-instanced baselines, our heuristics reduce the number of used machines often achieving the lower bound.

\section*{Acknowledgements}

We thank the anonymous reviewers for their comments on the manuscript.

This research is supported by a Polish National Science Center grant Opus (UMO-2017/25/B/ST6/00116).

\balance

\bibliographystyle{IEEEtran}
\bibliography{article}

\end{document}